\documentclass{article}
\usepackage{amsfonts}
\usepackage{bbm}
\usepackage{amsmath, amsthm, amssymb}
\usepackage{color,calc}
\usepackage{hyperref}

\newcommand{\sech}{\operatorname{sech}}

\newcommand{\dt}{\delta}
\newcommand{\ta}{\tau}

\newcommand{\et}{\eta}

\newcommand{\Om}{\Omega}
\newcommand{\ld}{\lambda}

\newtheorem{prop}{Proposition}
\newtheorem{lemm}{Lemma}
\newtheorem{thm}{Theorem}
\newtheorem{eg}{Example}

\begin{document}
%

\begin{center}{\Large \bf  A new $(\gamma_n,\sigma_k)-$ KP hierarchy and generalized dressing method }
\end{center}
\begin{center}
{\it Yuqin Yao$^{1)}$\footnote{Corresponding author:yyqinw@126.com
}, Yehui Huang$^{2)}$\footnote{huangyh@mails.tsinghua.edu.cn} and
Yunbo Zeng$^{3)}$\footnote{yzeng@math.tsinghua.edu.cn} }
\end{center}
\begin{center}{\small \it $^{1)}$Department of
  Applied Mathematics, China Agricultural University, Beijing, 100083, PR China\\
 $^{2)}$ School of Mathematics and Physics, North China Electric Power University, Beijing, 102206,
 China\\
$^{3)}$ Department of Mathematical Science, Tsinghua University,
Beijing, 100084 , PR China}
\end{center}

\vskip 12pt { \small\noindent\bf Abstract.}
 {A new $(\gamma_n,\sigma_k)-$ KP hierarchy with two new time series $\gamma_n$ and $\sigma_k$, which consists of $\gamma_n-$
 flow,
 $\sigma_k-$ flow and mixed $\gamma_n$ and $\sigma_k$ evolution equations of eigenfunctions, is proposed. Two reductions and constrained flows of
 $(\gamma_n,\sigma_k)-$ KP hierarchy are studied. The dressing method is generalized to the $(\gamma_n,\sigma_k)-$ KP hierarchy
 and some solutions are  presented.}


\section{Introduction}
Generalizations of KP hierarchy (KPH) attract a lot of interests
from both physical and mathematical points of
view\cite{date}-\cite{melni3}. One kind of generalization is the
multi-component KP hierarchy\cite{date}, which contains many
physical relevant nonlinear integrable systems such as
Davey-Stewartson equation, two-dimensional Toda lattice and
three-wave resonant integrable equations. Another kind of
generalization of KP equation is the so called KP equation with
self-consistent sources (KPESCS)\cite{melni1,melni2}. For example,
the first type and second type of KPESCS consists of KP equation
with some additional terms and eigenvalue problem or time evolution
equations of eigenfunctions of KP equation,
respectively\cite{melni1}-\cite{Hu2}.

 Denote the time series of KP hierarchy  by $\{t_n\}$. Recently, we
 proposed an approach to construct an extended KP hierarchy(exKPH)
 by introducing another time series $\{\tau_k\}$ \cite{liu2008,lin2008,yao2009}. The
 exKPH consists of $t_n-$ flow of KP hierarchy, $\tau_k-$ flow and
 the $t_n-$ evolution equations of eigenfunctions. To make difference, we may
 call the exKPH as $(t_n,\tau_k)-$KPH. The $(t_n,\tau_k)-$KPH
 contains the first type and second type of KPESCS. Also we
 developed the dressing method to solve the
 $(t_n,\tau_k)-$KPH\cite{liu2009}. \cite{mawx} generalized the  $(t_n,\tau_k)-$KPH
 to the $(\tau_n,\tau_k)-$KPH which consists of $\tau_n-$ flow, $\tau_k-$ flow and
 the $\tau_n-$ evolution equations of eigenfunctions and  $\tau_k-$ evolutions of
 eigenfunctions.
 However, \cite{mawx} didn't find the solution of
 $(\tau_n,\tau_k)-$KPH. In contrast to one $t_n-$evolution equation
 of eigenfunctions as coupling equation in our $(t_n,\tau_k)-$KPH,
 there are two coupling equations: $\tau_n-$ evolution and $\tau_k-$
 evolution equations of eigenfunctions in $(\tau_n,\tau_k)-$KPH. Our
 generalized dressing method can not be applied to the
 $(\tau_n,\tau_k)-$KPH due to too many coupling equations.

 In this
 paper, we generalize the $(t_n,\tau_k)-$KPH to
 $(\gamma_n,\sigma_k)-$KPH by introducing two new time series
 $\gamma_n$ and $\sigma_k$ with two parameters $\alpha_n$ and
 $\beta_k$. The $(\gamma_n,\sigma_k)-$KPH consists of $\gamma_n-$
 flow, $\sigma_k-$flow and one mixed $\gamma_n$ and $\sigma_k$
 evolution equation of eigenfunctions. The $(\gamma_n,\sigma_k)-$KPH
 can be reduced to the KPH and $(t_n,\tau_k)-$KPH, and contains
 first type and second type as well as mixed type of KPESCS as
 special cases. The constrained flows of the
 $(\gamma_n,\sigma_k)-$KPH can be regarded as generalization of
 Gelfand-Dickey hierarchy (GDH), which contains the first type,
 second type as well as mixed type of GDH with self-consistent
 sources in special cases. We also develop the dressing method to
 solve the $(\gamma_n,\sigma_k)-$KPH. Comparing with the
 multi-component generalization, we generalize the KPH by means of
 introducing two new time series and adding eigenfunctions as
 components.

 The paper is organized as follows: In section 2, we propose a new $(\gamma_n,\sigma_k)-$
 KPH. Section 3 presents the constrained flows of the $(\gamma_n,\sigma_k)-$
 KPH. Section 4 devotes to develop the generalized
 dressing method  for solving the $(\gamma_n,\sigma_k)-$
 KPH. Section 5 presents the N-soliton solutions and a
 conclusion is given in the last section.

\section{A new $(\gamma_n,\sigma_k)-$ KP hierarchy}
\subsection{The  KP hierarchy and extended  KP hierarchy}
Let us first recall the construction of the  KP hierarchy(KPH)
\cite{date,jimbo1,Sato,Date1982,dickey} and the extended KP
hierarchy(exKPH)\cite{liu2008,liu2009}. It is well known that the
pseudo-differential operator L with potential functions $u_i$ is
defined as
 $$L = \partial  + u_1\partial^{-1}+u_2\partial^{-2}+\cdots.$$
The KPH is given by\cite{dickey}
\begin{equation}
  \label{eqn:dKP-LaxEqn}
  L_{t_n}=[B_n,L],
\end{equation}
where $B_n=L^n_{+}$ stands for the differential part of $L^{n}$. The
compatibility of the $t_n-$ flow and    $t_k-$ flow of
(\ref{eqn:dKP-LaxEqn}) leads to the zero-curvature representation of
KPH
\begin{equation}
  \label{addzeroc}
 B_{n,t_k}- B_{k,t_n}+[B_n,B_k]=0.
\end{equation}
In particular, $B_2=\partial^2+u_1$,
$B_3=\partial^3+3u_1\partial+3(u_{1x}+u_2)$ and (\ref{addzeroc}) by
setting $t_2=y$, $t_3=t$ and $u_1=u$ yields the KP equation
$$(4u_t-12uu_x-u_{xxx})_x
-3u_{yy}=0.$$

Based on the observation that  the squared
eigenfunction symmetry constraint given by
$$
   L^{k}=B_k+\sum_{i=1}^Nq_i\partial^{-1}r_i,~~~~$$$$
    q_{i,t_n}=B_n(q_i),~
    r_{i,t_n}=-B_n^*(r_i),
 $$
is compatible with  KP hierarchy\cite{constrant1,constrant2},  we
proposed the exKPH  as follows in  \cite{liu2008}
\begin{subequations}
  \label{eqns:exdKP-LaxEqn}
  \begin{align}
    L_{t_n}&=[B_n,L],  \label{eqn:exdKP-LaxEqna}\\
    L_{\ta_k}&=[B_k+\sum_{i=1}^Nq_i\partial^{-1}r_i,L], \label{eqn:exdKP-LaxEqnb}\\
    q_{i,t_n}&=B_n(q_i),~
    r_{i,t_n}=-B_n^*(r_i),~i=1,\cdots,N. \label{eqn:exdKP-LaxEqnc}
  \end{align}
\end{subequations}
 The commutativity of (\ref{eqn:exdKP-LaxEqna}) and
 (\ref{eqn:exdKP-LaxEqnb}) under (\ref{eqn:exdKP-LaxEqnc}) gives rise to
 the following zero-curvature representation for exKPH (\ref{eqns:exdKP-LaxEqn})
\begin{subequations}
  \label{exkpzeroc}
  \begin{align}
   & B_{n,\tau_k}- B_{k,t_n}+[B_n,B_{k}]+[B_n,\sum_{i=1}^Nq_i\partial^{-1}r_i]_{+}=0,  \label{exkpzeroca}\\
   & q_{i,t_n}=B_n(q_i),~
    r_{i,t_n}=-B_n^*(r_i),~i=1,\cdots,N. \label{exkpzerocb}
  \end{align}
\end{subequations}
To different (\ref{eqns:exdKP-LaxEqn}) from (\ref{eqn:dKP-LaxEqn})
and the generalized KPH presented in this paper, we may denote
(\ref{eqns:exdKP-LaxEqn}) or (\ref{exkpzeroc}) by $(t_n,\tau_k)-$
KPH. We developed the dressing method to solve the $(t_n,\tau_k)-$
KPH and obtained its solutions in \cite{liu2009}. \cite{mawx}
generalized the $(t_n,\tau_k)-$ KPH to $(\tau_n,\tau_k)-$ KPH as
follows
\begin{subequations}
  \label{ma}
  \begin{align}
   & B_{n,\tau_k}- B_{k,\tau_n}+[B_n,B_{k}]+[B_n,\sum_{i=1}^Nq_i\partial^{-1}r_i]_{+}
   +[\sum_{i=1}^Nq_i\partial^{-1}r_i,B_k]_{+}=0,  \label{ma1}\\
   & q_{i,\tau_n}=B_n(q_i),~
    r_{i,\tau_n}=-B_n^*(r_i),\label{ma2}\\
    &q_{i,\tau_k}=B_k(q_i),~
    r_{i,\tau_k}=-B_k^*(r_i),~i=1,\cdots,N. \label{ma3}
  \end{align}
\end{subequations}
But \cite{mawx} didn't  find the solutions for the
$(\tau_n,\tau_k)-$ KPH (\ref{ma}). In contrast to one pair of
coupling equations (\ref{eqn:exdKP-LaxEqnc}) (or (\ref{exkpzerocb}))
in $(t_n,\tau_k)-$ KPH, there are two pairs of coupling equations
(\ref{ma2}) and (\ref{ma3}) in $(\tau_n,\tau_k)-$ KPH. In fact, the
dressing method developed in our paper\cite{liu2009} can not be
applied to the $(\tau_n,\tau_k)-$ KPH (\ref{ma}) since there are too
many (two) coupling systems (\ref{ma2}) and (\ref{ma3}).

\subsection{A new $(\gamma_n,\sigma_k)-$ KP hierarchy}

Stimulated by the  $(t_n,\tau_k)-$ KPH (\ref{eqns:exdKP-LaxEqn}) and
(\ref{exkpzeroc}), we propose the following generalized KPH with two
generalized time series $\gamma_n$ and $\sigma_k$:
\begin{subequations}
  \label{eqns:gkp3}
  \begin{align}
 & L_{\gamma_n}=[B_n+\alpha_n\sum_{i=1}^Nq_i\partial^{-1}r_i,L],\label{gkp3a}\\
&L_{\sigma_k}=[B_k+\beta_k\sum_{i=1}^Nq_i\partial^{-1}r_i,L],\label{gkp3b}\\
 &  \alpha_n(q_{i,\sigma_k}-B_k(q_i))-\beta_k(q_{i,\gamma_n}-B_n(q_i))=0,\nonumber\\
  &  \alpha_n(r_{i,\sigma_k}+B^{*}_k(r_i))-\beta_k(r_{i,\gamma_n}+B^{*}_n(r_i))=0,~~i=1,2,\cdots,N.\label{gkp3c}
   \end{align}
\end{subequations}
We will prove the compatibility of (\ref{gkp3a}) and (\ref{gkp3b})
under (\ref{gkp3c}) in the following theorem. First we need the
following Lemma presented in \cite{liu2008}
\begin{equation}
 \label{eqns:gkp8}
[B_n,\sum_{i=1}^Nq_i\partial^{-1}r_i]_-=\sum_{i=1}^NB_n(q_i)\partial^{-1}r_i-\sum_{i=1}^Nq_i\partial^{-1}B^{*}_n(r_i).
\end{equation}
\begin{thm}
  \label{thm:0}
The  $\gamma_n-$ flow (\ref{gkp3a}) and $\sigma_k-$ flow
(\ref{gkp3b}) under (\ref{gkp3c}) are compatible.
\end{thm}

\begin{proof}
Denote
$$
  \tilde{B}_n=B_n+\alpha_n\sum_{i=1}^Nq_i\partial^{-1}r_i,$$$$
    \tilde{B}_k=B_k+\beta_k\sum_{i=1}^Nq_i\partial^{-1}r_i.$$
In order to prove $L_{\gamma_n,\sigma_k}=L_{\sigma_k,\gamma_n}$,
i.e.
$$[\tilde{B}_{n,\sigma_k}-\tilde{B}_{k,\gamma_n}+[\tilde{B}_{n},\tilde{B}_{k}],L]=0$$
we only need to prove
\begin{equation}
 \label{eqns:gkp5}
\tilde{B}_{n,\sigma_k}-\tilde{B}_{k,\gamma_n}+[\tilde{B}_{n},\tilde{B}_{k}]=0.
\end{equation}

  For convenience, we omit $\sum$. We can find that
$$\tilde{B}_{n,\sigma_k}=B_{n,\sigma_k}+\alpha_n(q\partial^{-1}r)_{\sigma_k}=[B_k+\beta_k(q\partial^{-1}r),L^{n}]_+
+\alpha_n(q\partial^{-1}r)_{\sigma_k}$$\begin{equation}
 \label{eqns:gkp6}
=[B_k,L^{n}]_{+}+\beta_k[q\partial^{-1}r,L^{n}]_{+}+\alpha_nq_{\sigma_k}\partial^{-1}r+\alpha_nq\partial^{-1}r_{\sigma_k},
\end{equation}
and similarly,
\begin{equation}
 \label{eqns:gkp7}
\tilde{B}_{k,\gamma_n}=[B_n,L^{k}]_{+}+\alpha_n[q\partial^{-1}r,L^{k}]_{+}+\beta_kq_{\gamma_n}\partial^{-1}r+\beta_kq\partial^{-1}r_{\gamma_n}.
\end{equation}
Making use of the basic Lemma (\ref{eqns:gkp8}), we have
$$[\tilde{B}_{n},\tilde{B}_{k}]=[B_{n},B_{k}]+[B_{n},\beta_kq\partial^{-1}r]+[\alpha_nq\partial^{-1}r,B_{k}]=[L^{n}-(L^{n})_{-},
L^{k}-(L^{k})_{-}]_{+}$$
$$~~~~+\beta_k[B_{n},q\partial^{-1}r]_{+}+\alpha_n[q\partial^{-1}r,B_k]_{+}+\beta_k[B_{n},q\partial^{-1}r]_{-}+\alpha_n[q\partial^{-1}r,B_k]_{-}$$
$$=[B_n,L^{k}]_{+}+[L^{n},B_k]_{+}-[(L^{n})_{-},(L^{k})_{-}]_++\beta_k[B_{n},q\partial^{-1}r]_{+}+\alpha_n[q\partial^{-1}r,B_k]_{+}$$
\begin{equation}
 \label{eqns:gkp9}
+\beta_kB_n(q)\partial^{-1}r-\beta_kq\partial^{-1}B^{*}_n(r)-\alpha_nB_k(q)\partial^{-1}r+\alpha_nq\partial^{-1}B^{*}_k(r).
\end{equation}
Then (\ref{eqns:gkp6}), (\ref{eqns:gkp7}) and (\ref{eqns:gkp9})
under (\ref{gkp3c}) yields
$$\tilde{B}_{n,\sigma_k}-\tilde{B}_{k,\gamma_n}+[\tilde{B}_{n},\tilde{B}_{k}]=[\alpha_n(q_{\sigma_k}-B_k(q))-\beta_k(q_{\gamma_n}-B_n(q))]
\partial^{-1}r$$ $$+q\partial^{-1}[\alpha_n(r_{\sigma_k}+B^{*}_k(r))-\beta_k(r_{\gamma_n}+B^{*}_n(r))]=0.$$
\end{proof}

Then the compatibility of $\gamma_n-$ flow (\ref{gkp3a}) and
$\sigma_k-$ flow (\ref{gkp3b}) under (\ref{gkp3c})  gives rise to
the zero-curvature representation for (\ref{eqns:gkp3})
$$(B_n+\alpha_n\sum_{i=1}^Nq_i\partial^{-1}r_i)_{\sigma_k}-(B_k+\beta_k\sum_{i=1}^Nq_i\partial^{-1}r_i)_{\gamma_n}
$$$$
 +[B_n+\alpha_n\sum_{i=1}^Nq_i\partial^{-1}r_i,B_k+\beta_k\sum_{i=1}^Nq_i\partial^{-1}r_i]=0$$
 which under (\ref{gkp3c}) can be simplified as follows. Then we have
\begin{thm}
  \label{thm:1}
  The commutativity of (\ref{gkp3a}) and (\ref{gkp3b}) under
  (\ref{gkp3c}) gives rise to the zero-curvature
  equation for the generalized KPH with two generalized time series
\begin{subequations}
  \label{eqns:gkp4}
  \begin{align}
 &
 B_{n,\sigma_k}-B_{k,\gamma_n} +[B_n,B_k]+\beta_k[B_n,\sum_{i=1}^Nq_i\partial^{-1}r_i]_+
 +\alpha_n[\sum_{i=1}^Nq_i\partial^{-1}r_i,B_k]_+=0,\label{gkp4a}\\
 & \alpha_n(q_{i,\sigma_k}-B_k(q_i))-\beta_k(q_{i,\gamma_n}-B_n(q_i))=0,\nonumber\\
  & \alpha_n(r_{i,\sigma_k}+B^{*}_k(r_i))-\beta_k(r_{i,\gamma_n}+B^{*}_n(r_i))=0,~~i=1,2,\cdots,N,\label{gkp4b}
   \end{align}
\end{subequations}
with the Lax representation
\begin{equation}
 \label{laxp}
\psi_{\gamma_n}=(B_n+\alpha_n\sum_{i=1}^Nq_i\partial^{-1}r_i)(\psi),~\psi_{\sigma_k}=(B_k+\beta_k\sum_{i=1}^Nq_i\partial^{-1}r_i)(\psi).
\end{equation}
\end{thm}
We briefly call (\ref{eqns:gkp3}) and (\ref{eqns:gkp4}) as
$(\gamma_n,\sigma_k)-$ KPH. It is easy to see that
$(\gamma_n,\sigma_k)-$ KPH (\ref{eqns:gkp3})  and (\ref{eqns:gkp4})
for $\alpha_n=\beta_k=0$ reduces to KPH (\ref{eqn:dKP-LaxEqn}) and
(\ref{addzeroc}), $(\gamma_n,\tau_k)-$ KPH  for
$\alpha_n=0,~\beta_k=1$ reduces to $(t_n,\tau_k)-$KPH
(\ref{eqns:exdKP-LaxEqn}) and  (\ref{exkpzeroc}). So
$(\gamma_n,\sigma_k)-$ KPH (\ref{eqns:gkp3}) and (\ref{eqns:gkp4})
present a more generalized KPH which contains the KPH and
$(t_n,\tau_k)-$KPH  as the special cases.


\begin{eg}
Let us take $n=2$ and $k=3$, and set $\gamma_2=y,~\sigma_3=t$,
$u_1=u$. Then equations (\ref{eqns:gkp4}) becomes
\begin{subequations}
  \label{exam}
  \begin{align}
 &
 B_{2,t}-B_{3,y}+[B_2,B_3]+\beta_3[B_2,\sum_{i=1}^Nq_i\partial^{-1}r_i]_+
 +\alpha_2[\sum_{i=1}^Nq_i\partial^{-1}r_i,B_3]_+=0,\label{exama}\\
 &  \alpha_2(q_{i,t}-B_3(q_i))-\beta_3(q_{i,y}-B_2(q_i))=0,\nonumber\\
  & \alpha_2(r_{i,t}+B^{*}_3(r_i))-\beta_3(r_{i,y}+B^{*}_2(r_i))=0,~~i=1,2,\cdots,N,\label{examb}
   \end{align}
\end{subequations}
which gives the following nonlinear equation
\begin{subequations}
  \label{examequ}
  {\small\begin{align}
 &
4u_t-3\partial^{-1}u_{yy}-12uu_x-u_{xxx}-3\alpha_2\sum_{i=1}^N(q_ir_i)_y+4\beta_3\sum_{i=1}^N(q_ir_i)_x \nonumber\\
 &+3\alpha_2\sum_{i=1}^N(q_ir_{i,xx}-q_{i,xx}r_i)=0,\label{examequa}\\
 &  \alpha_2(q_{i,t}-q_{i,xxx}-3uq_{i,x}-\frac{3}{2}q_i\partial^{-1}u_y-\frac{3}{2}q_iu_x-\frac{3}{2}q_i\sum_{j=1}^Nq_jr_j)\nonumber\\
 &-\beta_3(q_{i,y}-q_{i,xx}-2uq_{i})=0,\nonumber\\
  & \alpha_2(r_{i,t}-r_{i,xxx}-3ur_{i,x}+\frac{3}{2}r_i\partial^{-1}u_y-\frac{3}{2}r_iu_x+\frac{3}{2}r_i\sum_{j=1}^Nq_jr_j) \label{examequb}\\
 & -\beta_3(r_{i,y}+r_{i,xx}+2ur_{i})=0,~~i=1,2,\cdots,N,\nonumber
   \end{align}}
\end{subequations}
with the Lax representation as follows
$$\psi_{y}=(\partial^{2}+2u+\alpha_2\sum_{i=1}^Nq_i\partial^{-1}r_i)(\psi),$$
\begin{equation}
\label{exam1lax}
\psi_{t}=(\partial^{3}+3u\partial+\frac{3}{2}\partial^{-1}u_y+\frac{3}{2}u_x+\frac{3}{2}\beta_3
\sum_{i=1}^Nq_i\partial^{-1}r_i)(\psi).\end{equation}
\end{eg}
Specially, when take $\alpha_2=\beta_3=0$; $\alpha_2=0,~\beta_3=1$;
$\alpha_2=1,~\beta_3=0$ and  $\alpha_2=1,~\beta_3=1$, respectively,
(\ref{examequ}) and  (\ref{exam1lax}) reduces to  the KP
equation\cite{dickey}, the first type of KP equation with
self-consistent sources\cite{melni1,melni2,zeng04}, the second type
of KP equation with self-consistent sources\cite{melni1,Hu2,liu2008}
and the mixed type of KP equation with self-consistent
sources\cite{Hu2} and their Lax representations, respectively.

\section{Reduction}

Consider the constraint given by
\begin{equation}
 \label{kconstraint1}L^{k}=B_k+\beta_k\sum_{i=1}^Nq_i\partial^{-1}r_i.
 \end{equation}
Then (\ref{gkp3b}) yields

\begin{equation}
 \label{addr1}(L^{k})_{\sigma_k}=[B_k+\beta_k\sum_{i=1}^Nq_i\partial^{-1}r_i,L^{k}]=0,
 \end{equation}
$$B_{k,\sigma_k}=(L^{k}_{\sigma_k})_+=0,$$
$$(\sum_{i=1}^Nq_i\partial^{-1}r_i)_{\sigma_k}=(L^{k}_{\sigma_k})_-=0,$$
which imply that $L$ , $B_k$,  $q_i$ and $r_i$  under
(\ref{kconstraint1}) are independent of $\sigma_k$. Subsequently,
$q_{i,\sigma_k}$ and $r_{i,\sigma_k}$ in (\ref{gkp3c}) should be
replaced by $\lambda_iq_i$ and $-\lambda_ir_i$ as in the case of
constrained flow of KP\cite{constrant1,constrant2}, namely
(\ref{gkp3c}) under the constraint (\ref{kconstraint1}) should be
replaced by

$$\alpha_n(\lambda_iq_{i}-B_k(q_i))-\beta_k(q_{i,\gamma_n}-B_n(q_i))=0,$$
\begin{equation}
\label{addr2}
\alpha_n(-\lambda_ir_{i}+B^{*}_k(r_i))-\beta_k(r_{i,\gamma_n}+B^{*}_n(r_i))=0.
\end{equation}
We will show that the constraint (\ref{kconstraint1}) is invariant
under the $\gamma_n-$ flow (\ref{gkp3a}) and (\ref{addr2}). In fact,
 making use
of  (\ref{gkp3a}),(\ref{eqns:gkp8}) and (\ref{addr2}), we have
$$(L^{k}-B_{k})_{\gamma_n}=(L^{k}_{\gamma_n})_-=[\tilde{B}_n,L^{k}]_-$$
$$(\beta_k\sum_{i=1}^Nq_i\partial^{-1}r_i)_{\gamma_n}=\beta_k\sum_{i=1}^N(q_{i,\gamma_n}\partial^{-1}r_i+q_i\partial^{-1}r_{i,\gamma_n})
=\sum_{i=1}^N[\beta_kB_n(q_i)\partial^{-1}r_i$$
$$+\alpha_n(\lambda_iq_{i}-B_k(q_i))\partial^{-1}r_i-\beta_kq_i\partial^{-1}B_n^{*}(r_i)+\alpha_nq_i\partial^{-1}(-\lambda_ir_{i}+B^{*}_k(r_i))]$$
$$=[B_n,\beta_k\sum_{i=1}^Nq_i\partial^{-1}r_i]_{-}-[B_k,\alpha_n\sum_{i=1}^Nq_i\partial^{-1}r_i]_-~~~~~~~~~~~~~~~~~~~~~~~~~~~~~~~
$$$$=[\tilde{B}_n,\beta_k\sum_{i=1}^Nq_i\partial^{-1}r_i]_--[B_k,\alpha_n\sum_{i=1}^Nq_i\partial^{-1}r_i]_-~~~~~~~~~~~~~~~~~~~~~~~~~~~~~~~
$$$$=[\tilde{B}_n,L^{k}]_--[\tilde{B}_n,B_k]_--[B_k,\tilde{B}_n]_-+[B_k,B_n]_-=[\tilde{B}_n,L^{k}]_-.~~~~~~~~~~~~~$$
Then
$$(L^{k}-B_{k}-\beta_k\sum_{i=1}^Nq_i\partial^{-1}r_i)_{\gamma_n}=0.$$
This means that the sub-manifold determined by the k-constraint
(\ref{kconstraint1}) is invariant under the $\gamma_n-$flow
(\ref{gkp3a}) and (\ref{addr2}).

Therefore, the constrained flow of  $(\gamma_n,\sigma_k)$-KPH
(\ref{eqns:gkp3}) and (\ref{eqns:gkp4}) under (\ref{kconstraint1})
reads
\begin{subequations}
  \label{constrainth}
  \begin{align}
 &
 B_{k,\gamma_n}
 +[B_k,B_n]+\beta_k[\sum_{i=1}^Nq_i\partial^{-1}r_i,B_n]_++\alpha_n[B_k,\sum_{i=1}^Nq_i\partial^{-1}r_i]_+=0,\label{constraintha}\\
 &  \alpha_n(\lambda_iq_{i}-B_k(q_i))-\beta_k(q_{i,\gamma_n}-B_n(q_i))=0,\nonumber\\
  &
  \alpha_n(-\lambda_ir_{i}+B^{*}_k(r_i))-\beta_k(r_{i,\gamma_n}+B^{*}_n(r_i))=0,~~i=1,2,\cdots,N.\label{constrainthc}\\
&with\nonumber\\
&B_n=(B_k+\beta_k\sum_{i=1}^Nq_i\partial^{-1}r_i)^{\frac{n}{k}}_{+}.
\end{align}
\end{subequations}

The system (\ref{constrainth}) can be regarded as the generalized
Gelfand-Dickey hierarchy (GDH). When $\alpha_n=\beta_k=0$,
(\ref{constrainth}) reduces to the GDH. When
$\alpha_n=1,~\beta_k=0$, (\ref{constrainth}) is just the first type
of GDH with self-consistent sources. When $\alpha_n=0,~\beta_k=1$,
(\ref{constrainth}) represent the second type of GDH with
self-consistent sources.

\begin{eg}
When $k=2,~n=3,~\gamma_3=t,~u_1=u$, (\ref{constrainth}) gives
\begin{subequations}
  \label{mixkdv}
  \begin{align}
 &
 u_t-\frac{1}{4}u_{xxx}-3uu_x+\alpha_3\sum_{i=1}^N(q_ir_i)_x+\frac{3}{4}\beta_2\sum_{i=1}^N(q_ir_{i,xx}-q_{i,xx}r_i)=0,\label{mixkdva}\\
 &  -\beta_2(q_{i,t}-q_{i,xxx}-3uq_{i,x}-\frac{3}{2}u_xq_i-\frac{3}{2}q_i\sum_{j=1}^Nq_jr_j)+\alpha_3(\lambda_{i}q_i-q_{i,xx}-2uq_i)=0,\label{mixkdvb}\\
  &  \beta_2(r_{i,t}-r_{i,xxx}-3ur_{i,x}-\frac{3}{2}u_xr_i+\frac{3}{2}r_i\sum_{j=1}^Nq_jr_j)-\alpha_3(-\lambda_{i}r_i+r_{i,xx}+2ur_i)=0,
  \nonumber\\
  &~~i=1,2,\cdots,N.\label{mixkdvc}
   \end{align}
\end{subequations}
which just is the mixed type of KdV equation with self-consistent
sources. (\ref{mixkdv}) with $\alpha_3=1,~\beta_2=0$ gives the first
type of KdV equation with sources\cite{kdv1,kdv2}. (\ref{mixkdv})
with $\alpha_3=0,~\beta_2=1$ gives the second  type of KdV equation
with sources\cite{liu2008}.
\end{eg}
\begin{eg}
When $k=3,~n=2$ and  $\gamma_2=t$, $u_1=u$, (\ref{constrainth})
gives rise to the mixed type of Boussinesq equation with
self-consistent sources
\begin{subequations}
  \label{mixbouss}
\small{  \begin{align}
 &
\frac{1}{3}u_{xxxx}+2(u^{2})_{xx}+u_{tt}+\sum_{i=1}^N[-\frac{4}{3}\beta_3(q_ir_i)_{xx}+\alpha_2(q_ir_i)_{xt}+\alpha_2
(q_{i,xx}r_i-q_ir_{i,xx})_x]=0,\label{mixkdva}\\
 &
 \alpha_2[\lambda_{i}q_i-q_{i,xxx}-3uq_{i,x}-q_i(\frac{3}{2}\partial^{-1}u_y+\frac{3}{2}u_x+\frac{3}{2}\sum_{i=1}^Nq_jr_j)]\nonumber\\
& -\beta_3(q_{i,t}-q_{i,xx}-2uq_i)=0,\nonumber\\
  &
  \alpha_2[-\lambda_{i}r_i-r_{i,xxx}-3ur_{i,x}+r_i(\frac{3}{2}\partial^{-1}u_y-\frac{3}{2}u_x+\frac{3}{2}\sum_{i=1}^Nq_jr_j)]\nonumber\\
 & - \beta_3(r_{i,t}-r_{i,xx}-2ur_i)=0,\nonumber\\
  &~~i=1,2,\cdots,N.\label{mixkdvc}
   \end{align}}
\end{subequations}
(\ref{mixbouss}) with $\alpha_2=1,~\beta_3=0$ gives the first type
of Boussinesq equation with sources. (\ref{mixkdv}) with
$\alpha_2=0,~\beta_3=1$ gives the second  type of Boussinesq
equation with sources\cite{liu2008}.

\end{eg}

\section{Dressing approach for $(\gamma_n,\sigma_k)$-KPH}
Inspired by Refs\cite{dickey,os96}, we consider the generalized
dressing approach for $(\gamma_n,\sigma_k)$-KPH. Assume that
operator $L$ of $(\gamma_n,\sigma_k)$-KPH  can be written as a
dressing form
\begin{equation}
  \label{eqn:dress}
  L=W\partial W^{-1},
\end{equation}
\begin{equation}
  \label{w}
  W =1 + w_1\partial^{-1} + w_2\partial^{-2}+\cdots .
\end{equation}

\begin{prop}
\label{eqn:add} If $W$ defined by (\ref{w}) satisfies
\begin{subequations}
\label{prop1}
\begin{align}
 W_{\gamma_n}=-L^{n}_{-}W+\alpha_n\sum_{i=1}^Nq_i\partial^{-1}r_iW,\label{prop1a}\\
W_{\sigma_k}=-L^{k}_{-}W+\beta_k\sum_{i=1}^Nq_i\partial^{-1}r_iW\label{prop1b}
\end{align}
\end{subequations}
 then $L$ satisfies
(\ref{gkp3a}) and (\ref{gkp3b}).
\end{prop}

\begin{proof}
Based on (\ref{eqn:dress}) and (\ref{prop1a}), we have
$$L_{\gamma_n}=W_{\gamma_n}\partial W-W\partial W^{-1}W_{\gamma_n}W^{-1}$$
$$=(L^{n}_{+}+\alpha_n\sum_{i=1}^Nq_i\partial^{-1}r_i)L-L(L^{n}_{+}+\alpha_n\sum_{i=1}^Nq_i\partial^{-1}r_i)$$
$$=[B_n+\alpha_n\sum_{i=1}^Nq_i\partial^{-1}r_i,L].$$
Similarly, we can prove that $L$ satisfies (\ref{gkp3b}).
\end{proof}
It is well known that the Wronskian determinant\cite{dickey}
$$Wr(h_1,\cdots,h_N)=\left|
    \begin{matrix}
      h_1 & h_2 & \cdots & h_N\\
      h'_1 &  h'_2 & \cdots &  h'_N\\
      \vdots & \vdots & \vdots & \vdots\\
      h^{N-1}_1  & h^{N-1}_2 & \cdots &h^{N-1}_N
    \end{matrix}\right|$$
is a $\tau-$ function of the KPH and the Nth order differential
operator given by
\begin{equation}
  \label{woperator}
W= \frac{1}{Wr(h_1,\cdots,h_N)}\left|
    \begin{matrix}
      h_1 & h_2 & \cdots & h_N &1\\
      h'_1 &  h'_2 & \cdots &  h'_N&\partial\\
      \vdots & \vdots & \vdots & \vdots\\
      h^{N}_1  & h^{N}_2 & \cdots &h^{N}_N&\partial^{N}
    \end{matrix}\right|
\end{equation}
provides the dressing operator,
 where $h_1,~h_2,~\cdots, h_N$ are $N$ independent
functions and satisfy $W(h_i)=0.$

 This dressing operator $W$
is constructed as follows: Let $f_i$, $g_i$ satisfy
\begin{subequations}
  \label{eq:fg}
  \begin{align}
    &f_{i,\gamma_n}=\partial^{n}(f_i),\quad f_{i,\sigma_k}=\partial^{k}(f_i)\quad \\
    &g_{i,\gamma_n}=\partial^n(g_i),\quad
    g_{i,\sigma_k}=\partial^k(g_i),~i=1,\ldots,N,
  \end{align}
\end{subequations}
and let $h_i$ be the linear combination of $f_i$ and $g_i$
\begin{equation}
  \label{eq:h'}
  h_i=f_i+F_i(\alpha_n\gamma_n+\beta_k\sigma_k)g_i\quad i=1,\ldots,N,
\end{equation}
with $F_i(X)$ being a differentiable function of $X$,
$X=\alpha_n\gamma_n+\beta_k\sigma_k$.

Define
\begin{equation}
  \label{eqn:EigenFns}
 \small{ q_i=-\dot{F}_iW(g_i),~
  r_i=(-1)^{N-i}\frac{Wr(h_1,\cdots,\hat{h}_i,\cdots,h_N)}
  {Wr(h_1,\cdots,h_N)}, i=1,\ldots, N}
\end{equation}
where the hat $\hat{\;}$ means rule out this term from the Wronskian
determinant, $\dot{F}_i=\frac{d\alpha_i}{dX}$.
We have
\begin{thm}
  \label{thm}
  Let $W$ be defined by (\ref{woperator}) and (\ref{eq:h'}), $L=W\partial W^{-1}$, $q_{i}$ and $r_{i}$
  be given by (\ref{eqn:EigenFns}), then $W$, $L$,
  $q_i$, $r_i$ satisfy (\ref{prop1}) and  $(\gamma_n,\sigma_k)$-KPH (\ref{eqns:gkp3}) and (\ref{eqns:gkp4}).
\end{thm}

To prove Theorem \ref{thm}, we need several lemmas. The first one is
given by  Oevel and Strampp\cite{os96}:
\begin{lemm}
  \label{lm:OS}
  $W^{-1}=\sum_{i=1}^N h_i\partial^{-1} r_i.$
\end{lemm}

\begin{lemm}
  \label{lm:key}\cite{liu2008}
  The operator $\partial^{-1} r_i W$ is a non-negative differential operator
  and
  \begin{equation}
    \label{eq:key-eq}
    (\partial^{-1}r_i W)(h_j)=\dt_{ij},~1\le i,j\le N.
  \end{equation}
\end{lemm}

{\sl Proof of Theorem \ref{thm}.}
   For (\ref{prop1a}), taking $\partial_{\gamma_n}$ to the identity
  $W(h_i)=0$, using \eqref{eq:fg}, \eqref{eq:h'}, the definition
  (\ref{eqn:EigenFns}) and Lemma \ref{lm:key}, we find
  \begin{align*}
    0=&(W_{\gamma_n})(h_i)+(W\partial^n)(f_i)+\alpha_n\dot F_i
    W(g_i)+F_i(W\partial^n)(g_i)\\
    =&(W_{\gamma_n})(h_i)+(W\partial^n)(h_i)-\alpha_nq_i\\
    =&(W_{\gamma_n})(h_i)+(L^nW)(h_i)-\alpha_n\sum_{j=1}^Nq_j\dt_{ji}\\
    =&(W_{\gamma_n}+L^n_-W-\alpha_n\sum_{j=1}^Nq_j\partial^{-1}r_jW)(h_i).
  \end{align*}
  Since the non-negative difference operator acting on $h_i$ in the last
  expression has degree $<N$, it can not annihilate $N$ independent functions
  unless the operator itself vanishes. Hence(\ref{prop1a}) is
  proved. Similarly, we can prove (\ref{prop1b}). Then Proposition \ref{eqn:add} leads to (\ref{gkp3a}) and (\ref{gkp3b}).

  The first equation in (\ref{gkp3c}) is easy to be verified by a direct calculation, so it remains
  to prove the second equation in (\ref{gkp3c}). Firstly, we see that
  $$(W^{-1})_{\gamma_n}=-W^{-1}W_{\gamma_n}W^{-1}=-W^{-1}(L^n_+-L^n+\alpha_n\sum_{j=1}^Nq_j\partial^{-1}r_j)$$
  \begin{equation}
  \label{wtau1}
  =\partial^nW^{-1}-W^{-1}B_n-\alpha_nW^{-1}\sum_{j=1}^Nq_j\partial^{-1}r_j.
  \end{equation}
\begin{equation}
  \label{wsigam1}
  (W^{-1})_{\sigma_k}=\partial^kW^{-1}-W^{-1}B_k-\beta_kW^{-1}\sum_{j=1}^Nq_j\partial^{-1}r_j.
  \end{equation}
On the other hand, from  $W^{-1}=\sum
  h_i\partial^{-1}r_i$  we have
\begin{equation}
  \label{wtau2}
    (W^{-1})_{\gamma_n}=\sum \partial^n(h_i)\partial^{-1}r_i+\sum h_i\partial^{-1}r_{i,\gamma_n}
  \end{equation}
\begin{equation}
  \label{wsigam2}
    (W^{-1})_{\sigma_k}=\sum \partial^k(h_i)\partial^{-1}r_i+\sum h_i\partial^{-1}r_{i,\sigma_k}
  \end{equation}
It is obviously that
$\alpha_n(\ref{wsigam1})-\beta_k(\ref{wtau1})=\alpha_n(\ref{wsigam2})-\beta_k(\ref{wtau2})$,
i.e.
$$-\beta_k\sum \partial^n(h_i)\partial^{-1}r_i-\beta_k\sum h_i\partial^{-1}r_{i,\gamma_n}+\alpha_n\sum \partial^k(h_i)\partial^{-1}r_i
+\alpha_n\sum
h_i\partial^{-1}r_{i,\sigma_k}$$$$=-\beta_k(\partial^nW^{-1}-W^{-1}B_n)_{-}+\alpha_n(\partial^kW^{-1}-W^{-1}B_k)_{-}$$
$$=-\beta_k\sum \partial^n(h_i)\partial^{-1}r_i+\beta_k\sum h_i\partial^{-1}B_n^{*}(r_i)+
\alpha_n\sum \partial^k(h_i)\partial^{-1}r_i-\alpha_n\sum
h_i\partial^{-1}B_k^{*}(r_i)$$ The above equations gives
$$\beta_k\sum h_i\partial^{-1}(r_{i,\gamma_n}+B_n^{*}(r_i))-\alpha_n\sum h_i\partial^{-1}(r_{i,\sigma_k}+B_k^{*}(r_i))=0,$$
which implies the second equation in (\ref{gkp3c}) holds.

\section{N-soliton solutions for $(\gamma_n,\sigma_k)$-KPH}
Using Theorem \ref{thm}, we can find N-soliton solutions to every
equations in the $(\gamma_n,\sigma_k)$-KPH (\ref{eqns:gkp3}) and
(\ref{eqns:gkp4}). Let us illustrate it by solving (\ref{examequ}).
We take the solution of (\ref{eq:fg}) as follows
\begin{displaymath}
 f_i:=\exp(\ld_i x+\ld_i^{2} y+\ld_i^{3} t)=e^{\xi_i}, \quad
 g_i:=\exp(\mu_ix+\mu_i^{2}y+\mu_i^{3}t)=e^{\et_i}
\end{displaymath}
\begin{equation}\label{eq:36}
  h_i:=f_i+F_i(\alpha_2y+\beta_3t)g_i=2\sqrt{F_i}\exp(\frac{\xi_i+\et_i}{2})\cosh(\Om_i),
  ~\Om_i=\frac12(\xi_i-\et_i-\ln F_i).
\end{equation}

For example, when $N=1$, $W=\partial-\frac{h'}{h}$,
$$L=W\partial^{-1}W=\partial+\frac{(\lambda_1-\mu_1)^{2}}{4}\sech^{2}\Omega_{1}\partial^{-1}+\cdots$$
The one-soliton solution for (\ref{examequ}) with $N=1$ as follows
$$u=\frac{(\lambda_1-\mu_1)^{2}}{4}\sech^{2}\Omega_{1}$$
$$q_1=\sqrt{\alpha_2F_{1y}+\beta_3F_{1t}}(\lambda_1-\mu_1)e^{\xi_1+\eta_1}\sech\Omega_{1}$$
$$r_1=\frac{1}{\sqrt{F_1}}e^{-(\xi_1+\eta_1)}\sech\Omega_{1}$$
In the case of $N=2$, the two-soliton solution for (\ref{examequ})
is given
$$u=\partial^{2}\ln \Theta,$$
$$q_1=(\alpha_2F_{1y}+\beta_3F_{1t})\frac{(\lambda_1-\mu_1)(\lambda_2-\mu_1)}{\Theta}(1+F_2\frac{(\lambda_1-\mu_2)(\mu_2-\mu_1)}
{(\lambda_1-\lambda_2)(\lambda_2-\mu_1)}e^{\eta_2-\xi_1})e^{\eta_1},$$
$$q_2=(\alpha_2F_{2y}+\beta_3F_{2t})\frac{(\lambda_2-\mu_2)(\lambda_1-\mu_2)}{\Theta}(1+F_1\frac{(\lambda_2-\mu_1)(\mu_1-\mu_2)}
{(\lambda_2-\lambda_1)(\lambda_1-\mu_2)}e^{\eta_1-\xi_2})e^{\eta_2},$$
$$r_1=\frac{1+F_2e^{\eta_2-\xi_2}}{\lambda_2-\mu_1}e^{-\xi_1},~r_2=\frac{1+F_1e^{\eta_1-\xi_1}}{\lambda_2-\mu_1}e^{-\xi_2}$$
where
$$\Theta=1+F_1\frac{\lambda_2-\mu_1}{\lambda_2-\lambda_1}e^{\eta_1-\xi_1}+F_2\frac{\mu_2-\lambda_1}{\lambda_2-\lambda_1}e^{\eta_2-\xi_2}
+F_1F_2\frac{\mu_2-\mu_1}{\lambda_2-\lambda_1}e^{\eta_1+\eta_2-\xi_1-\xi_2}.$$

\section{Conclusion}
In contrast to the multi-component generalization of KP hierarchy,
we generalize KP hierarchy by introducing new time series $\gamma_n$
and $\sigma_k$ and adding eigenfunctions as components. The
$(\gamma_n,\sigma_k)-$KPH includes KP hierarchy and extended KP
hierarchy, and contains first type and second type as well as mixed
type of KP equation with self-consistent sources as special cases.
The constrained flows of $(\gamma_n,\sigma_k)-$KPH can be regarded
as the generalized Gelfand-Dickey hierarchy. We develop the dressing
method for solving the $(\gamma_n,\sigma_k)-$KPH and present its
N-soliton solutions.

\section*{Acknowledgement}
This work is supported by National Basic Research Program of China
(973 Program) (2007CB814800),  National Natural Science Foundation
of China (10901090,10801083) and  Chinese Universities Scientific
Fund (2011JS041).

\end{document}